\documentclass[runningheads,a4paper]{llncs}
\usepackage[T1]{fontenc}

\newcommand{\edgecapacity}[1]{c_{#1}}

\newcommand{\nodecapacity}[1]{B_{#1}}

\newcommand{\source}{{s}}

\newcommand{\sink}{{t}}

\mathchardef\mhyphen="2D

\bibstyle{plainurl}

\usepackage{authblk}
\usepackage{cite}
\usepackage{graphicx}
\usepackage{calc}
\usepackage{subfig}
\usepackage{color}
\usepackage{url}
\usepackage{enumerate}
\usepackage{mathtools}
\usepackage{amssymb}
\usepackage{amsfonts}

\newtheorem{remark*}{Remark}

\newtheorem{openprob*}{Open Problem}
\newtheorem{example*}{Example}
\newtheorem{note*}{Note}
\newtheorem{prob*}{Problem}

\usepackage[ruled,vlined,nofillcomment,longend,linesnumbered]{algorithm2e}
\newlength\algowd

\allowdisplaybreaks
\begin{document}
\mainmatter
\title{Temporal flows in Temporal networks \thanks{This work was partially supported by (i) the School of EEE and CS and the NeST initiative of the University of Liverpool, (ii) the NSERC Discovery grant, (iii)  the Polish National Science Center grant DEC-2011/02/A/ST6/00201, and (iv)  the FET EU IP Project MULTIPLEX under contract No. 317532.},\thanks{To appear in the 10th International Conference on Algorithms and Complexity (CIAC 2017)}
}
\author{Eleni C. Akrida\inst{1} \and Jurek Czyzowicz\inst{2} \and Leszek G\k{a}sieniec\inst{1} \and \\{\L}ukasz Kuszner\inst{3} \and Paul G. Spirakis\inst{1,4}}
\authorrunning{Akrida, Czyzowicz, G\k{a}sieniec, Kuszner, Spirakis}
\institute{Department of Computer Science, University of Liverpool, UK\\
\email{\{Eleni.Akrida2,L.A.Gasieniec,P.Spirakis\}@liverpool.ac.uk}
\and
Universit\'e du Qu\'ebec en Outaouais, Dep. d'Informatique, Gatineau, QC, Canada\\
\email{jurek@uqo.ca}
\and
Gda\'nsk University of Technology, Faculty of Electronics, Telecommunications and Informatics, Poland\\
\email{kuszner@eti.pg.gda.pl}
\and
Computer Technology Institute \& Press ``Diophantus'' (CTI), Patras, Greece
}
\maketitle

\begin{abstract}
We introduce temporal flows on temporal networks~\cite{kempe,spirakis}, i.e., networks the links of which exist only at certain moments of time. Such networks are ephemeral in the sense that no link exists after some time. Our flow model is new and differs from the ``flows over time'' model, also called ``dynamic flows'' in the literature. We show that the problem of finding the maximum amount of flow that can pass from a source vertex $\source$ to a sink vertex $\sink$ up to a given time is solvable in Polynomial time, even when node buffers are bounded. We then examine mainly the case of unbounded node buffers. We provide a simplified static \emph{Time-Extended network} ($\mathrm{STEG}$), which is of \emph{polynomial size to the input} and whose static flow rates are equivalent to the respective temporal flow of the temporal network; using $\mathrm{STEG}$, we prove that the maximum temporal flow is equal to the minimum \emph{temporal $\source \mhyphen \sink$ cut}. We further show that temporal flows can always be decomposed into flows, each of which moves only through a journey, i.e., a directed path whose successive edges have strictly increasing moments of existence. We partially characterise networks with random edge availabilities that tend to eliminate the $\source \to \sink$ temporal flow. We then consider \emph{mixed} temporal networks, which have some edges with specified availabilities and some edges with random availabilities; we show that it is \pmb{\#P}-hard to compute the \emph{tails and expectations of the maximum temporal flow} (which is now a random variable) in a mixed temporal network.
\end{abstract}

\section{Introduction and motivation}

\subsection{Our model and the problem}

It is generally accepted to describe a network topology using a graph, whose vertices represent the communicating entities and edges correspond to the communication opportunities between them. Consider a directed graph (network) $G(V,E)$ with a set $V$ of $n$ vertices (nodes) and a set $E$ of $m$ edges (links). Let $\source, \sink \in V$ be two special vertices called the \emph{source} and the \emph{sink}, respectively; for simplicity, assume that no edge enters the source $\source$ and no edge leaves the sink $\sink$. We also assume that an infinite amount of a quantity, say, a liquid, is available in $\source$ at time zero. However, our network is \emph{ephemeral}; each edge is available for use only at certain \emph{days} in time, described by positive integers, and after some (finite) day in time, no edge becomes available again. For example, some edge $e=(u,v)$ may exist only at days $5$ and $8$; the reader may think of these days as instances of availability of that edge. Our liquid, located initially at node $\source$, can flow in this ephemeral network through edges only at days at which the edges are available.

Each edge $e\in E$ in the network is also equipped with a \emph{capacity} $\edgecapacity{e}>0$ which is a positive integer, unless otherwise specified. We also consider each node $v\in V$ to have an internal buffer (storage) $B(v)$ of maximum size $\nodecapacity{v}$; here, $\nodecapacity{v}$ is also a positive integer; initially, we shall consider both the case where $\nodecapacity{v}= +\infty$, for all $v \in V$, and the case where all nodes have finite buffers. From Section~\ref{sec:teg} on, we only consider unbounded (infinite) buffers.

The \emph{semantics} of the flow of our liquid within $G$ are the following:
\begin{itemize}
\item Let an amount $x_v$ of liquid be at node $v$, i.e., in $B(v)$, at the \emph{beginning} of day $l$, for some $l \in \mathbb{N}$. Let $e=(v,w)$ be an edge that exists at day $l$. Then, $v$ may \emph{push} some of the amount $x_v$ through $e$ at day $l$, as long as that amount is at most $\edgecapacity{e}$. This quantity will arrive to $w$ at the \emph{end} of \emph{the same day}, $l$, and will be stored in $B(w)$.
\item At the end of day $l$, for any node $w$, some flows may arrive from edges $(v,w)$ that were available at day $l$. Since each such quantity of liquid has to be stored in $w$, the sum of all flows incoming to $w$ plus the amount of liquid that is already in $w$ at the end of day $l$, after $w$ has sent any flow out of it at the beginning of day $l$, must not exceed $\nodecapacity{w}$.
\item Flow arriving at $w$ at (the end of) day $l$ can leave $w$ only via edges existing at days $l' >l$.
\end{itemize}

Thus, our flows are not flow rates, but flow amounts (similar to considerations in \emph{transshipment problems}).

Notice that we assume above that we have absolute knowledge of the days of existence of each edge. This information is detailed, but it can model a range of scenarios where a network is operated by many users and detailed description of link existence (or lack thereof) is needed; for example, one may need to have detailed information on planned maintenance on pipe-sections in a water network to assure restoration of the network services, and one may need to know in advance the time schedule of a rail network to circulate passengers. However, such a detailed input can not be used in all practical cases; often, instead of having a specific list of days of existence of some edge(s), one may be able to obtain statistical knowledge of a pattern of existence of connections via previously gathered information. A model that captures such cases is the model of \emph{Mixed Temporal Networks}, which we introduce and study here, along with the traditional Temporal Networks model.

We provide efficient solutions to the \emph{Maximum temporal flow problem} (MTF): Given a directed graph $G$ with edge availabilities, distinguished nodes $\source, \sink$, edge capacities and node buffers as previously described, and also given a specific day $l' > 0$, find the maximum value of the quantity of liquid that can arrive to $\sink$ by (the end of) day $l'$.

Notice that no flow will arrive to $\sink$ in fewer days than the ``temporal distance of $\sink$ from $\source$'' (the smallest \emph{arrival time} of any $\source \to \sink$ path with strictly increasing days of availability on its consecutive edges; here, \emph{arrival time} is the day of availability of the last edge on the path).\\

\noindent\emph{\textbf{Relation to previously studied problems.}}~~
MTF is related both to the problem of standard maximum (instantaneous) flows and to the {\em transshipment problem}; in the latter, the network has several sources and sinks, each source with a specified supply and each sink with a specified demand. Studies on the quickest transshipment problem~\cite{hoppe_trans,kamiyama} consider the problem of sending exactly the right amount of flow out of each source and into each sink in the minimum overall time; the networks considered there have edges with transit times. Other authors have also considered problems such as minimising capacity violations in transshipment networks~\cite{radzik}, where the initial capacity constraints render the problem infeasible, but an increase in the capacities by some additive terms (the \emph{capacity violations}) allow a feasible shipment so as to minimise an objective function.

However, MTF is very different from both the standard maximum flow problem and the transshipment problem. Indeed, in the network of Figure~\ref{fig:example_difference} with all node buffers and edge capacities being infinite, but \emph{all edges existing only at the same day}, say $l=5$, no flow can \emph{ever} arrive to $\sink$.
\begin{figure}[htb]
\centering\includegraphics[scale=0.44]{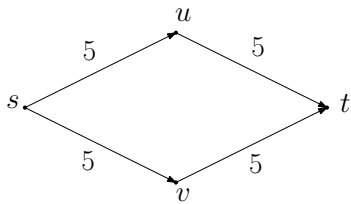}
\caption{Difference between temporal flows and standard flows.}
\label{fig:example_difference}
\end{figure}

Moreover, MTF has not been a well examined problem in previous work on (continuous or discrete) \emph{dynamic flows} considered in \cite{skutella5,skutella1,tardos,skutella4,hoppe}, and references therein. Indeed, the ``transit time'' on each edge of our networks is less than one day, and \emph{only} if the edge exists at that day. \emph{All} units of flow that are located at the tail of an edge at a moment when the edge becomes available may pass through the edge all together (like a ``packet'' of information), if the edge capacity allows it. In fact, our model is an extreme case of a version of a discrete dynamic flows model called Dynamic Dynamic Flows~\cite[Chapter~8]{hoppe-phd}.

Also, in our model, the existence of node buffers (holdover flow) is \emph{necessary}; \emph{in contrast to all previous flow and transshipment studies}, our networks cannot propagate flow without holdover flows, i.e., node buffers storing flow units.

So, we consider here ephemeral networks that change over time, as well as flows that are dynamic and the movement of which is determined by the temporal structure of the network.

\subsection{Previous work}

The traditional (static) network flows were extensively studied in the seminal book of Ford and Fulkerson~\cite{ford-fulkerson} (see also Ahuja et al~\cite{ahuja}) and the relevant literature is vast. \emph{Dynamic network flows} (see, e.g.,~\cite{hoppe-phd}) refer to \emph{static} directed networks, the edges of which have capacities as well as transit times. Ford and Fulkerson~\cite{ford-fulkerson} formulated and solved the dynamic maximum flow problem. For excellent surveys on dynamic network flows, the reader is also referred to the work of Aronson~\cite{aronson}, the work of Powell~\cite{powell}, and the great survey by Skutella~\cite{skutella3}. Dynamic network flows are also called \emph{flows over time}. In \cite{skutella1}, the authors review \emph{continuous} flows over time where $f_e(\theta)$ is the rate of flow (per time unit) entering edge $e$ at time $\theta$; the values of $f_e(\theta)$ are assumed to be Lebesgue-measurable functions. In our model, we assume that any flow amount that can pass through an edge at an instant of existence, will pass, i.e., our $f_e(\theta)$ is infinite in a sense. In a technical report~\cite{tjandra}, the authors examine earliest arrival flows with time-dependent travel times and edge capacities; they describe the flow equations of their model and give their own Ford-Fulkerson approach and dynamic cut definitions; although different to their model, our work gives an intuitively simpler definition of a temporal cut. For various problems on flows over time, see~\cite{skutella5,skutella1,ghaffari,skutella4,hoppe,Kappmeier15,woeginger,skutella2}. Flows over time have been also considered in problems of scheduling jobs in a network~\cite{boland}.

Classical static flows have recently been re-examined for the purpose of approximating their maximum value or improving their time complexity~\cite{radzik2,radzik3,serna,batra_garg,wiese,andoni, naveen2, naveen1, racke-flow, sankowski, madry, orlin}. Network flows have also been used in multi-line addressing~\cite{eisenbrand}.

Another relevant problem to the one we consider here is the \emph{transshipment problem}. In a transshipment problem, shipments of products (i.e., of amounts of products, in analogy to amounts of flows in our model) are allowed between source-sink pairs in a network, where each source has some supply and each sink has some demand. In some applications, shipments may also be allowed between sources and between sinks. Transshipment problems have also been extensively studied in literature; for example, in studies on the \emph{quickest} transshipment problem~\cite{hoppe_trans,kamiyama}, the authors consider networks with transit times on their edges and study the problem of sending exactly the right amount of flow out of each source and into each sink in the minimum overall time. Other authors have considered problems such as minimising capacity violations in transshipment networks~\cite{radzik}, where the initial capacity constraints render the problem infeasible, but an increase in the capacities by some additive terms (the \emph{capacity violations}) allow a feasible shipment so as to minimise an objective function.

Temporal networks, defined by Kempe et al.~\cite{kempe}, are graphs \emph{the edges of which exist only at certain instants of time, called labels} (see also~\cite{spirakis}). So, they are a type of \emph{dynamic} networks. Various aspects of temporal (and other dynamic) networks were also considered in the work of Erlebach et al~\cite{erlebach} and in~\cite{akrida-spaa, akrida-waoa, akrida-algo, avin, casteigts, dutta, michail, nicosia, o'dell, sch}; as far as we know, this is the first work to examine flows on temporal networks. Berman~\cite{berman} proposed a similar model to temporal networks, called \emph{scheduled networks}, in which each edge has separate departure and arrival times; he showed that the max-flow min-cut theorem holds in scheduled networks, when edges have \emph{unit capacities}. There is also literature on models of temporal networks with random edge availabilities~\cite{chaintreau,clementi,akrida-jpdc}, but to the best of our knowledge, ours is the first work on flows in such temporal networks.

Perhaps the closest model in the flows literature to the one we consider is the ``\emph{Dynamic}\footnote{The first ``dynamic'' term refers to the dynamic nature of the underlying graph, i.e., appearance and disappearance of its edges} dynamic network flows'', studied by Hoppe in his PhD thesis~\cite[Chapter 8]{hoppe-phd}. In~\cite[Chapter 8]{hoppe-phd}, Hoppe introduces \emph{mortal edges} that exist between a start and an end time; still, Hoppe assumes transmission rates on the edges and the ability to hold any amount of flow on a node (infinite node buffers). Thus, our model is an extreme case of the latter, since we assume that edges exist only at specific days (instants) and that our transit rates are virtually unbounded, since at one instant \emph{any amount} of flow can be sent through an edge if the capacity allows.

\subsection{Our results}
We introduce flows in Temporal Networks for the first time. We are interested in the maximum total amount of flow that can pass from $\source$ to $\sink$ during the lifetime of the network; notice that the edges of the network exist only at some days during the lifetime, different in general for each edge.

In Section~\ref{sec:def}, we formulate the problem of computing the maximum temporal flow and in Section~\ref{sec:LP}, we show that it can be solved in polynomial time, even when the node capacities are finite. This is in contrast to the NP-hardness result conjectured by Hoppe~\cite[personal communication with Klinz]{hoppe-phd} for bounded holdover flows in \emph{dynamic dynamic} networks, which is the model closest to ours.

The remainder of the paper mainly concerns networks in which the nodes have unbounded buffers, i.e., buffers with infinite capacity. In Section~\ref{sec:time_ext}, we define the corresponding time-extended network ($\mathrm{TEG}$) which converts our problem to a static flow problem (following the time-extended network tradition in the literature~\cite{ford-fulkerson}). However, we manage to simplify $\mathrm{TEG}$ into a \emph{simplified time-extended network} ($\mathrm{STEG}$), the size of which, i.e., number of nodes and edges, is \emph{polynomial on the input}, and not exponential as usual in flows over time. Using the $\mathrm{STEG}$, we prove our \emph{maximum temporal flow-minimum temporal cut} theorem; temporal cuts extend the traditional cut notion, since the edges included in a cut need not exist at the same day(s) in time.
We also show that temporal flows are always decomposable into a set of flows, each moving through a particular journey, i.e., directed path whose time existence of successive edges strictly increases. 

Admittedly, the encoding of the input in our temporal network problems is quite detailed but as previously mentioned, specific description of the edge availabilities may be required in a range of network infrastructure settings where there is a planned schedule of link existence. On the positive side, some problems that are weakly NP-hard in similar dynamic flow models become polynomially solvable in our model. However, in many practical scenarios it is reasonable to assume that not all edge availabilities are known in advance, e.g., in a water network where there may be unplanned disruptions at one or more pipe sections; in these cases, one may have statistical information on the pattern of link availabilities. In Section~\ref{sec:random_journeys}, we demonstrate cases of temporal flow networks with randomly chosen edge availabilities that eliminate the flow that arrives at $\sink$ asymptotically almost surely. We also introduce and study flows in mixed temporal networks for the first time; these are networks in which the availabilities of some edges are random and the availabilities of some other edges are specified. In such networks, the value of the maximum temporal flow is a random variable. Consider, for example, the temporal flow network of Figure~\ref{fig:mixed_net} where there are $n$ directed disjoint two-edge paths from $\source$ to $\sink$. Assume that \emph{every} edge independently selects a \emph{unique} label uniformly at random from the set $\{1,\ldots, \alpha\},~\alpha \in \mathbb{N}^*$. The edge capacities are the numbers drawn in the boxes, with $w_i' \geq w_i$ for all $i$. Here, the value of the maximum $\source \to \sink$ flow is a random variable that is the sum of Bernoulli random variables. This already indicates that the exact calculation of the maximum flow in mixed networks is a hard problem. In Section~\ref{sec:mixed_net} we show for mixed networks that it is \pmb{\#P}-hard to compute tails and expectations of the maximum temporal flow.
\begin{figure}[tbh]
\centering
\includegraphics[scale=0.4]{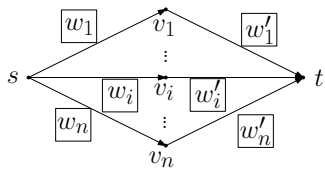}
\caption{A mixed temporal network}
\label{fig:mixed_net}
\end{figure}

\subsection{Formal Definitions}\label{sec:def}
\begin{definition}[(Directed) Temporal Graph]
Let $G=(V,E)$ be a directed graph. A (directed) temporal graph on $G$ is an ordered triple $G(L)=(V,E,L)$, where $L=\{L_e \subseteq \mathbb{N}:e\in E\}$ assigns a \emph{finite} set $L_e$ of discrete labels to every edge (arc) $e$ of $G$. $L$ is called the \emph{labelling} of $G$. The labels, $L_e$, of an edge $e \in E$ are the \emph{integer time instances (e.g., days)} at which $e$ is available.
\end{definition}

\begin{definition}[Time edge]
Let $e=(u,v)$ be an edge of the underlying digraph of a temporal graph and consider a label $l\in L_e$. The ordered triplet $(u,v,l)$, also denoted as $(e,l)$, is called \emph{time edge}. We denote the set of time edges of a temporal graph $G(L)$ by $E_L$.
\end{definition}

A basic assumption that we follow here is that when a (flow) entity passes through an available edge $e$ at time $t$, then it can pass through a subsequent edge only at some time $t'\geq t+1$ and only at a time at which that edge is available. In the tradition of assigning ``transit times'' in the dynamic flows literature, one may think that any edge $e$ of the graph has some \emph{transit time}, $tt_e$, with $0 < tt_e < 1$, but \emph{otherwise arbitrary and not specified}. Henceforth, we will use $tt_e=0.5$ for all edges $e$, without loss of generality in our results; any value of $tt_e$ between 0 and 1 will lead to the same results in our paper.

\begin{definition}[Journey]
A \emph{journey} from a vertex $u$ to a vertex $v$, denoted as \emph{$u \to v$ journey}, is a sequence of time edges $(u, u_1, l_1)$, $(u_1, u_2, l_2)$, $\ldots$ , $(u_{k-1}, v, l_k)$, such that $l_i < l_{i +1}$, for each $1 \leq i \leq k - 1$. The last time label, $l_k$, is called the \emph{arrival time} of the journey.
\end{definition}

\begin{definition}[Foremost journey]
A $u \to v$ journey in a temporal graph is called \emph{foremost journey} if its arrival time is the minimum arrival time of all $u \to v$ journeys' arrival times, under the labels assigned to the underlying graph's edges. We call this arrival time the \emph{temporal distance}, $\delta(u,v)$, of $v$ from $u$.
\end{definition}
Thus, no flow arrives to $\sink$ (starting from $\source$) on or before any time $l< \delta(\source,\sink)$.

\begin{definition}[Temporal Flow Network]
A temporal flow network $\big( G(L),\source, \sink, c, B \big)$ is a temporal graph $G(L)=(V,E,L)$ equipped with:
\begin{enumerate}
\item a source vertex $\source$ and a sink (target) vertex $\sink$
\item for each edge $e$, a capacity $\edgecapacity{e}>0$; usually the capacities are assumed to be integers.
\item for each node $v$, a buffer $B(v)$ of storage capacity $\nodecapacity{v}>0$; $\nodecapacity{\source}$ and $\nodecapacity{\sink}$ are assumed to be infinite.
\end{enumerate}
If all node capacities are infinite, we denote the temporal flow network by $\big( G(L),\source, \sink, c\big)$.
\end{definition}

\begin{definition}[Temporal Flows in Temporal Flow Networks]\label{def:flow}
Let $\big( G(L)=(V,E,L),\source, \sink, c, B \big)$ be a temporal flow network. Let:
\begin{eqnarray*}
\delta_u^+ & = & \{e\in E| \exists w \in V, e=(u,w)\} \\
\delta_u^- & = & \{e\in E| \exists w \in V, e=(w,u)\}
\end{eqnarray*}
be the outgoing and incoming edges to $u$. Also, let $L_R(u)$ be the set of labels on all edges incident to $u$ along with an extra label $0$ (artificial label for initialization), i.e.,
\[ L_R(u) = \bigcup_{e\in \delta_u^+ \cup \delta_u^-} L_e \cup \{0\}\]

A temporal flow on $G(L)$ consists of a non-negative real number $f(e,l)$ for each time-edge $(e,l)$, and real numbers $b_u^-(l), b_u^\mu (l), b_u^+(l)$ for each node $u \in V$ and each ``day'' $l$. These numbers must satisfy all of the following:
\begin{enumerate}
\item \label{item-1} $0 \leq f(e,l) \leq \edgecapacity{e}$, for every time edge $(e,l)$,
\item $0 \leq b_u^-(l) \leq \nodecapacity{u},~0 \leq b_u^\mu(l) \leq \nodecapacity{u},~0 \leq b_u^+(l) \leq \nodecapacity{u}$, for every node $u$ and every $l \in L_R(u)$
\item for every $e \in E$, $f(e,0)=0$,
\item for every $v \in V\setminus \{\source\}$, $b_v^-(0) = b_v^\mu(0) = b_v^+(0) = 0$,
\item for every $e \in E$ and $l \not\in L_e$, $f(e,l)=0$,
\item at time $0$ there is an infinite amount of flow ``units'' available at the source $\source$,
\item for every $v \in V\setminus \{\source\}$ and for every $l \in L$, $b_v^-(l) = b_v^+(l_{prev})$, where $l_{prev}$ is the largest label in $L_R(v)$ that is smaller than $l$,
\item (Flow out on day $l$) for every $v \in V\setminus \{\source\}$ and for every $l$, $b_v^\mu(l) =b_v^-(l) - \sum_{e \in \delta_v^+} f(e,l)$,
\item \label{item-9} (Flow in on day $l$) for every $v \in V\setminus \{\source\}$ and for every $l$, $b_v^+(l) =b_v^\mu(l) + \sum_{e \in \delta_v^-} f(e,l)$.
\end{enumerate}
\end{definition}

\begin{note*}
One may think of $b_v^-(l), b_v^\mu(l), b_v^+(l)$ as the buffer content of liquid in $v$ at the ``morning'',``noon'', i.e., after the departures of flow from $v$, and ``evening'', i.e., after the arrivals of flow to $v$, of day $l$.
\end{note*}

\begin{note*}
For a temporal flow $f$ on an acyclic $G(L)$, if one could guess the (real) numbers $f(e,l)$ for each time-edge $(e,l)$, then the numbers $b_v^-(l), b_v^\mu(l), b_v^+(l)$, for every $v \in V$, can be computed by a single pass over an order of the vertices of $G(L)$ from $\source$ to $\sink$. This can be done by following (\ref{item-1}) through (\ref{item-9}) from Definition \ref{def:flow} from $\source$ to $\sink$.
\end{note*}

\begin{definition}[Value of a Temporal Flow]
The value $v(f)$ of a temporal flow $f$ is $b_\sink^+(l_{max})$ under $f$, i.e., the amount of liquid that, via $f$, reaches $\sink$ during the lifetime of the network ($l_{max}$ is the maximum label in $L$). If $b_\sink^+(l_{max})>0$ for a particular flow $f$, we say that $f$ is \emph{feasible}.
\end{definition}

\begin{definition}[Mixed temporal networks]\label{def:mixed}
Given a directed graph $G=(V,E)$ with a source $\source$ and a sink $\sink$ in $V$, let $E=E_1 \cup E_2$, so that $E_1 \cap E_2 = \emptyset$, and:
\begin{enumerate}
\item the labels (availabilities) of edges in $E_1$ are specified, and
\item each of the labels of the edges in $E_2$ is drawn uniformly at random from the set $\{1,2, \ldots, \alpha\}$, for some even integer $\alpha$\footnote{We choose an even integer to simplify the calculations in the remainder of the paper. However, with careful adjustments, the results would still hold for an arbitrary integer.}, independently of the others.
\end{enumerate}
We call such a network \emph{``Mixed Temporal Network $[1,\alpha]$''} and denote it by $G(E_1,E_2,\alpha)$.
\end{definition}

Note that (traditional) temporal networks as previously defined are a special case of the mixed temporal networks, in which $E_2=\emptyset$. However, with some edges being available at random times, the value of a temporal flow (until time $\alpha$) becomes a random variable and the study of relevant problems requires a different approach than the one needed for (traditional) temporal networks.

\begin{prob*}[Maximum Temporal Flow (MTF)]
Given a temporal flow network $\big( G(L),\source, \sink, c, B \big)$ and a day $d \in \mathbb{N}^*$, compute the maximum $b_\sink^+(d)$ over all flows $f$ in the network.
\end{prob*}

\section{LP for the MTF problem with or without bounded buffers}\label{sec:LP}

In the description of the MTF problem, if $d$ is not a label in $L$, it is enough to compute the maximum $b_\sink^+(l_m)$ over all flows, where $l_m$ is the maximum label in $L$ that is smaller than $d$. Henceforth, we assume $d=l_{max}$ unless otherwise specified; notice that the analysis does not change: if $d<l_{max}$, one can remove all time-edges with labels larger than $b$ and solve MTF in the resulting network with new maximum label at most $d$.

Note also that $b_\sink^+(l_{max})$ is not necessarily equal to the total outgoing flow from $\source$ during the lifetime of the network\footnote{The total outgoing flow from $\source$ by some day $x$ is the sum of all flow amounts that have ``left'' $\source$ by day $x$: $\sum_{l \in L_R(\source) \setminus \{ l^* \in \mathbb{N} : l*>x \}}\sum_{e \in \delta_\source^+}f(e,l)$.} , where the lifetime is $l_{max}-l_{min}$, $l_{min}$ being the smallest label in the network. For example, consider the network of Figure~\ref{fig:maxflow_vs_outflow}, where the labels of an edge are the numbers written next to it and its capacity is the number written inside the box; for $d=5$, the \emph{maximum flow by day $5$ is $b_\sink^+(5) = 8$}, i.e., the flow where $5$ units follow the journey $\source \to v \to \sink$ and $3$ units follow the journey $\source \to u \to v \to \sink$; however, the total outgoing flow from $\source$ by day $5$ is $10>8$.

\begin{figure}[htb]
\centering\includegraphics[scale=0.45]{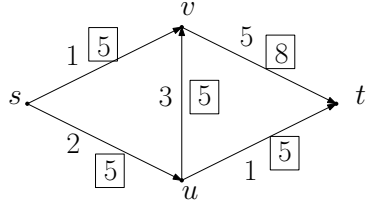}
\caption{Outgoing flow from $\source$ is not always the same as maximum flow by some day $d$; here $d=5$.}
\label{fig:maxflow_vs_outflow}
\end{figure}

Let $\Sigma$ be the set of conditions of Definition~\ref{def:flow}. The optimization problem, $\Pi$:
\[
\left\{\begin{array}{ll}
\text{max (over all $f$) } & b_\sink^+(d)\\
\text{subject to }& \Sigma
\end{array} \right\}
\]
is a \emph{linear program} with unknown variables $\{f(e,l), b_v^-(l), b_v^+(l)\},~\forall l \in L, \forall v \in V$, since each condition in $\Sigma$ is either a linear equation or a linear inequality in the unknown variables. Therefore, by noticing that the number of equations and inequalities are polynomial in the size of the input of $\Pi$, we get the following Lemma:

\begin{lemma}\label{lem:LP}
Maximum Temporal Flow is in P, i.e., can be solved in polynomial time in the size of the input, even when the node buffers are finite, i.e., bounded.
\end{lemma}

\begin{note*}
Recall that $E_L$ denotes the set of time edges of a temporal graph. If $n=|V|, m=|E|$ and $k=|E_L|=\sum_e |L_e|$, then MTF can be solved in sequential time polynomial in $n+m+k$  when the capacities and buffer sizes can be represented with polynomial in $n$ number of bits. In the remainder of the paper, we shall investigate more efficient approaches for MTF.
\end{note*}

\begin{note*}
Lemma \ref{lem:LP} for bounded node buffers is in wide contrast with the claim that the corresponding problem in dynamic dynamic network flows is NP-complete~\cite[p.~82]{hoppe-phd}.
\end{note*}

\section{Temporal Networks with unbounded buffers at nodes}\label{sec:teg}
\subsection{Basic remarks}
We consider here the MTF problem for temporal networks on underlying graphs with $\nodecapacity{v}=+\infty,~\forall v \in V$.

\begin{definition}[Temporal Cut]
Let $\big( G(L),\source, \sink, c \big)$ be a temporal flow network on a digraph $G$. A set of time-edges, $S$, is called a temporal cut (separating $\source$ and $\sink$) if the removal from the network of $S$ results in a temporal flow network with \emph{no $\source \to \sink$ journey}.
\end{definition}

\begin{definition}[Minimal Temporal Cut]
A set of time-edges, $S$, is called a minimal temporal cut (separating $\source$ and $\sink$) if:
\begin{enumerate}
\item it is a temporal cut, and
\item the removal from the network of any $S' \subset S$ results in a temporal flow network with at least one journey from $\source$ to $\sink$, i.e., any proper subset of $S$ is \emph{not} a temporal cut.
\end{enumerate}
\end{definition}

\begin{definition}
Let $S$ be a temporal cut of $\big( G(L)=(V,E,L),\source, \sink, c \big)$. The \emph{capacity} of the cut is $c(S):= \sum_{(e,l) \in S} c(e,l)$, where $c(e,l) = \edgecapacity{e},~\forall l$.
\end{definition}

In Figure~\ref{fig:cut_capacity}, the numbers next to the edges are their availability labels and the numbers in the boxes are the edge capacities; here, a minimal temporal cut is $S=\{ \big( (\source, v), 1 \big), \big( (\source, v), 7 \big) \}$ with capacity $c(S)=20$. Notice that another minimal cut is $S'=\{\big( (v,t),8\big) \}$ with capacity $c(S')=2$.

\begin{figure}[htb]
\centering\includegraphics[scale=0.5]{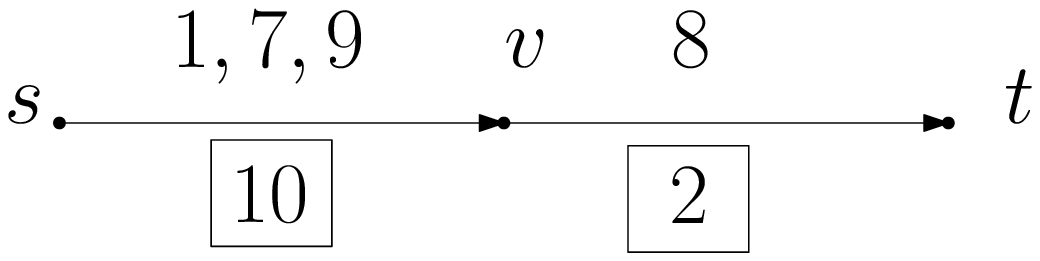}
\caption{$S=\{ \big( (\source, v), 1 \big), \big( (\source, v), 7 \big) \}$ is a minimal cut.}
\label{fig:cut_capacity}
\end{figure}

It follows from the definition of a temporal cut:
\begin{lemma}
Let $S$ be a (minimal) temporal cut in $\big( G(L)=(V,E,L),\source, \sink, c \big)$. If we remove $S$ from $G(L)$, no flow can ever arrive to $\sink$ during the lifetime of $G(L)$.
\end{lemma}
\begin{proof}
The removal of $S$ leaves no $\source \to \sink$ journey and any flow from $\source$ needs at least one journey to reach $\sink$, by definition.
\end{proof}

\subsection{The time-extended flow network and its simplification}\label{sec:time_ext}
Let $\big( G(L)=(V,E,L),\source, \sink, c \big)$ be a temporal flow network on a directed graph $G$. Let $E_L$ be the set of time edges of $G(L)$. Following the tradition in literature~\cite{ford-fulkerson}, we construct from $G(L)$ a \emph{static} flow network called \emph{time-extended} that corresponds to $G(L)$, denoted by $\mathrm{TEG}(L)=(V^*, E^*)$. By construction, $\mathrm{TEG}(L)$ admits the same maximum flow as $G(L)$. $\mathrm{TEG}(L)$ is constructed as follows.

For every vertex $v\in V$ and for every time step $i=0,1,\ldots,l_{max}$, we add to $V^*$ a copy, $v_i$, of $v$. $V^*$ also contains a copy of $v$ for every time edge $(x,v,l)$ of $G(L)$; in particular, we consider a copy $v_{l+tt}$ of $v$ in $V^*$, for some $l \in \mathbb{N}$, if $(x,v,l)\in E_L$, for some $x\in V$. 
Notice that $0 < tt < 1$ (by definition of the transit times), so if a vertex $v\in V$ has an incoming edge $e$ with label $l$ and an outgoing edge with label $l+1$, the copies $v_{l+tt},v_{l+1}$ of $v$ in $V^*$ will never be identical (see Figure~\ref{vertex-copies-fig}).
\begin{figure}[htb]
\centering\includegraphics[scale=0.5]{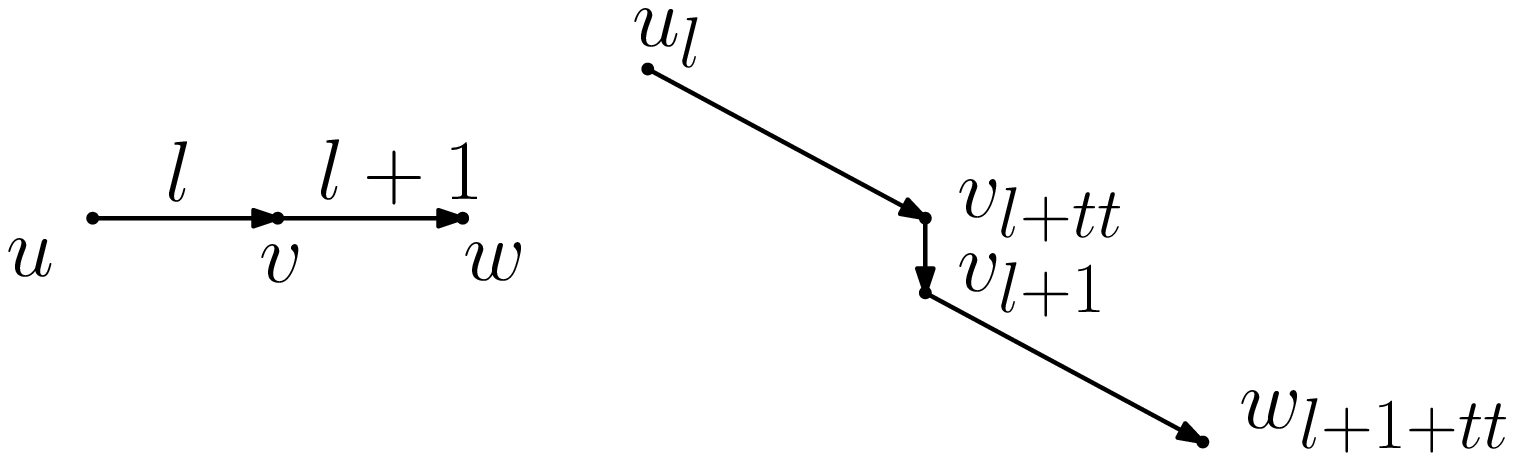}
\caption{The copies of vertex $v$ in $\mathrm{TEG}(L)$.}
\label{vertex-copies-fig}
\end{figure}

$E^*$ has a directed edge (called \emph{vertical}) from a copy of vertex $v$ to the \emph{next} copy of $v$, for any $v \in V$. More specifically,
\[
\forall v\in V,~(v_i,v_j) \in E^* \iff  \begin{cases} v_i,v_j \in V^*, & \mbox{ and} \\     j>i, & \mbox{ and} \\  \forall k > i : v_k \in V^* \implies k \geq j &\end{cases}
\]
Furthermore, for every time edge of $G(L)$, $E^*$ has a directed edge (called \emph{crossing}) as follows:
\[
\forall u,v \in V, l\in \mathbb{N},~(u,v,l)\in E \iff (u_l,v_{l+tt}) \in E^*
\]
Every crossing edge $e \in \mathrm{TEG}(L)$ that connects copies of vertices $u,v \in V$ has the capacity of the edge $(u,v) \in G(L)$, $\edgecapacity{e} = \edgecapacity{u,v}$. Every vertical edge $e \in \mathrm{TEG}(L)$ has capacity $\edgecapacity{e} = \nodecapacity{v} = +\infty$. The source and target vertices in $\mathrm{TEG}(L)$ are the first copy of $\source$ and the last copy of $\sink$ in $V^*$, respectively. Note that $|V^*| \leq |V| \cdot l_{max} +  |E_L|$ and $|E^*| \leq |V| \cdot l_{max} + 2 |E_L|$. 

We will now ``simplify'' $\mathrm{TEG}(L)$ as follows: we convert vertical edges between consecutive copies of the same vertex into a \emph{single vertical edge (with infinite capacity)} from the first to the last copy in the sequence and we remove all intermediate copies; we only perform this simplification when no intermediate node is an endpoint of a crossing edge. We call the resulting network \emph{simplified time-extended} network and we denote it by $\mathrm{STEG}(L)=(V',E')$.

In particular, for every vertex $v\in V$, $V^*$ has a copy $v_0$ of $v$, and a copy for each time edge that includes $v$ either as a first or as a last endpoint.
We consider a copy $v_l$ of $v$ in $V'$ iff $(v,x,l)\in E_L$, for some $x\in V$.
we consider a copy $v_{l+tt}$ of $v$ in $V^*$ iff $(x,v,l)\in E_L$, for some $x\in V$.

$E'$ has a directed \emph{vertical} edge from a copy of vertex $v$ to the \emph{next} copy of $v$, for any $v \in V$. More specifically,
\[
\forall v\in V,~(v_i,v_j) \in E' \iff  \begin{cases} v_i,v_j \in V', & \mbox{ and} \\     j>i, & \mbox{ and} \\  \forall k > i : v_k \in V' \implies k \geq j &\end{cases}
\]
Furthermore, for every time edge of $G(L)$, we consider the \emph{crossing} edge as in the time-extended graph, i.e.:
\[
\forall u,v \in V, l\in \mathbb{N},~(u,v,l)\in E \iff (u_l,v_{l+tt}) \in E'
\]

Every crossing edge $e \in \mathrm{STEG}(L)$, i.e., every edge that connects copies of different vertices $u,v \in V$, has the capacity of the edge $(u,v) \in G(L)$, $\edgecapacity{e} = \edgecapacity{u,v}$. Every edge $e \in \mathrm{STEG}(L)$ between copies of the same vertex $v \in V$ has capacity $\edgecapacity{e} = \nodecapacity{v} = +infty$. The source and target vertices in $\mathrm{STEG}(L)$ are the first copy of $\source$ and the last copy of $\sink$ in $V'$ respectively. Note that $|V'| \leq |V| + 2 |E_L|$ and $|E'| \leq |V|+ 3 |E_L|$.

Denote the first copy of any vertex $v \in V$ in the time-extended network by $v_{copy_0}$, the second copy by $v_{copy_1}$, the third copy by $v_{copy_2}$, etc. Let also:
\begin{eqnarray*}
\delta_u^+ & = & \{e\in E| \exists w \in V, e=(u,w)\} \\
\delta_u^- & = & \{e\in E| \exists w \in V, e=(w,u)\}
\end{eqnarray*}
An $\source \to \sink$ flow $f$ in $G(L)$ defines an $\source \to \sink$ flow (rate), $f_R$, in the time-extended network $\mathrm{STEG}(L)$ as follows:
\begin{itemize}\setlength\itemsep{1em}
\item The flow from the first copy of $\source$ to the next copy is the sum of all flow units that ``leave'' $\source$ in $G(L)$ throughout the time the network exists:
\[
f(\source_{copy_0},\source_{copy_1}) := \sum_{l\in \mathbb{N}} \sum_{e \in \delta_\source^+} f(e,l)
\]

\item The flow from the first copy of any \emph{other} vertex to the next copy is zero:
\[
\forall v\in V\setminus \source, ~f(v_{copy_0},v_{copy_1}) := 0
\]

\item The flow on any crossing edge that connects some copy $u_l$ of vertex $u\in V$ and the copy $v_{l+tt}$ of some other vertex $v\in V$ is exactly the flow on the time edge $(u,v,l)$:
\[
\forall (u_l,v_{l+tt})\in E', ~f(u_l,v_{l+tt}) := f((u,v),l)
\]

\item The flow between two \emph{consecutive} copies $v_x$ and $v_y$, for some $x,y$, of the same vertex $v \in V$ corresponds to the units of flow stored in $v$ from time $x$ up to time $y$ and is the difference between the flow \emph{received} at the first copy through all incoming edges and the flow \emph{sent} from the first copy through all outgoing crossing edges. So, $\forall v\in V, i=1,2, \ldots$, it is:
\[
f(v_{copy_i},v_{copy_{i+1}}) :=  \textstyle{\sum}_{z \in V'} f(z,v_{copy_i}) - \textstyle{\sum}_{u \in V' \setminus v_{copy_{i+1}}} f(v_{copy_i},u)
\]
\end{itemize}

\paragraph{Example.} 
Figure~\ref{temporal-flow-fig} shows a temporal network $G(L)$ with source $\source$ and sink $\sink$. The labels of an edge are shown next to the edge and the capacity of an edge is shown \emph{written in a box} next to the edge. The respective simplified time-extended static graph $\mathrm{STEG}(L)$ is shown in Figure~\ref{extended-time-flow-fig}. The capacity of an edge is shown \emph{written in a box} next to the edge. Notice that edges between copies of the same vertex have infinite capacities (equal to the infinite capacity of the vertex buffer) which are not shown in the figure.
\begin{figure}[!ht]
    \subfloat[Temporal flow network $G(L)$\label{temporal-flow-fig}]{%
      \includegraphics[width=0.45\textwidth]{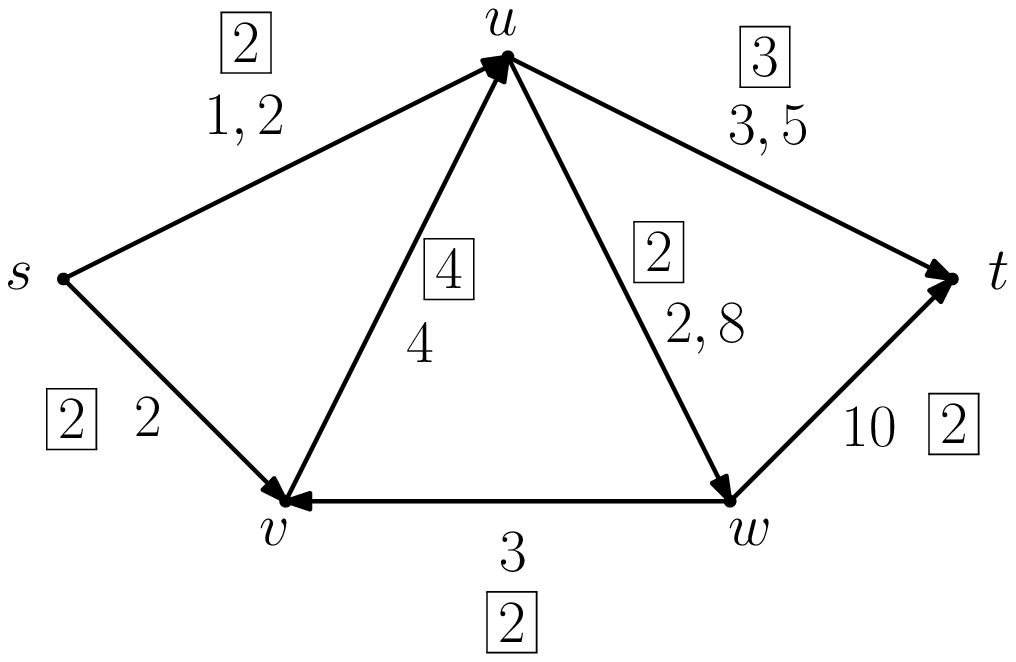}
    }
    \hfill
    \subfloat[Simplified time extended network $\mathrm{STEG}(L)$\label{extended-time-flow-fig}]{%
      \includegraphics[width=0.45\textwidth]{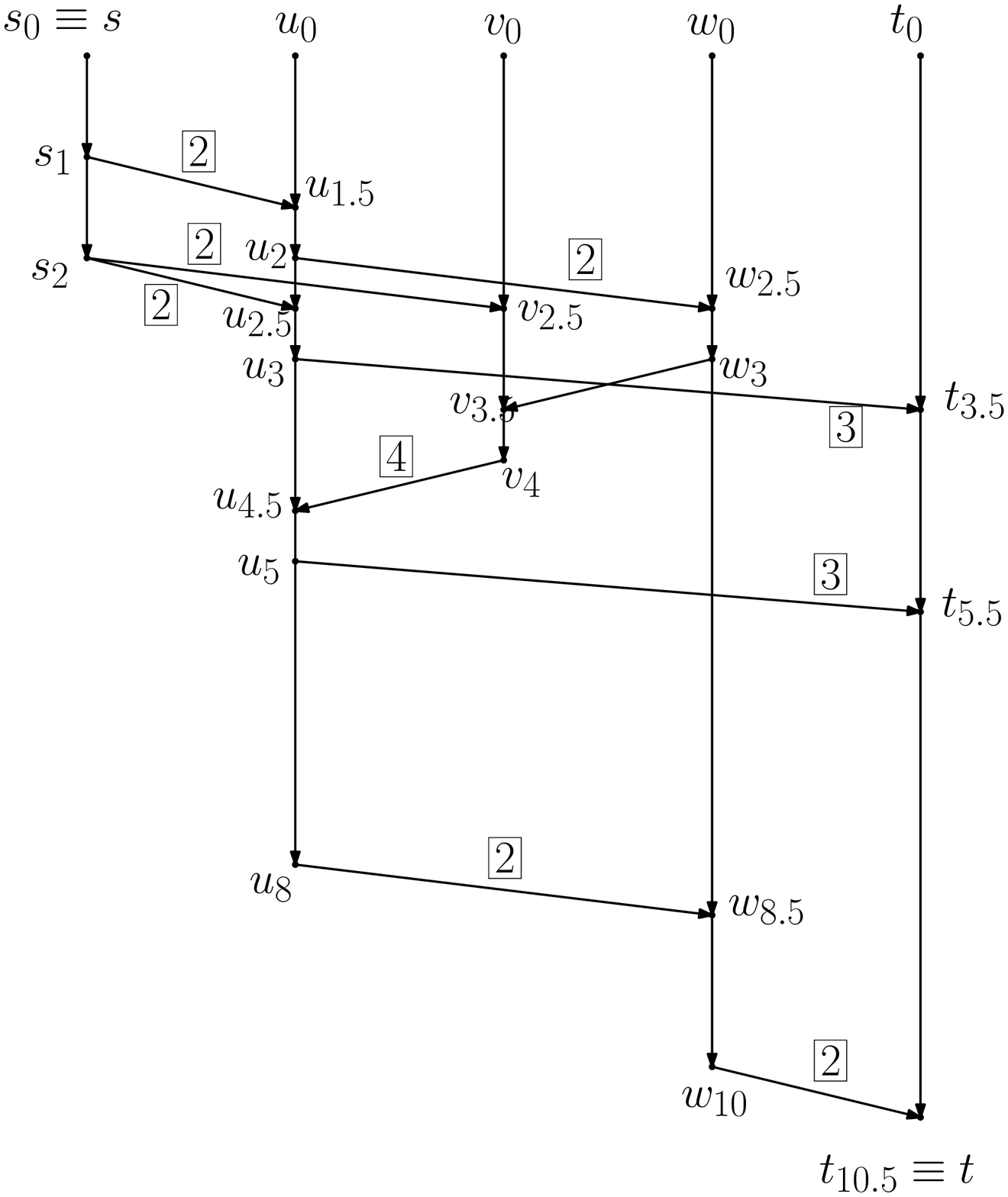}
    }
    \caption{Constructing the Simplified time-extended network}
		\label{time-extended-graph-flow-fig}
\end{figure}

Let $f_R$ be a static flow rate in the static network $\mathrm{STEG}(L)$ that corresponds to a temporal flow $f$ in a temporal flow network $\big( G(L)=(V,E,L),\source, \sink, c \big)$. By the construction of $\mathrm{STEG}(L)$, it follows:
\begin{lemma}\label{lem:steg_flow}
Given a temporal flow network $\big( G(L)=(V,E,L),\source, \sink, c \big)$ on a directed graph $G$,
\begin{enumerate}
\item The maximum temporal flow (from $\source$ to $\sink$), $max_f v(f)$, in $G(L)$ is equal to the maximum (standard) flow rate from $\source$ to $\sink$ in the \emph{static} network $\mathrm{STEG}(L)$.
\item A temporal flow $f$ is proper in $G(L)$ (i.e., satisfies all constraints) iff its corresponding static flow rate $f_R$ is feasible in $\mathrm{STEG}(L)$.
\end{enumerate}
\end{lemma}

\begin{lemma}\label{lem:min_cut}
The minimum capacity $\source \mhyphen \sink$ cut of the static network $\mathrm{TEG}(L)$ is equal to the minimum capacity $\source \mhyphen \sink$ cut of the static network $\mathrm{STEG}(G)$.
\end{lemma}
\begin{proof}
Any minimum capacity cut in either $\mathrm{TEG}(L)$ or $\mathrm{STEG}(L)$ uses crossing edges. But the crossing edges are the same in both networks. Therefore, the lemma holds.
\end{proof}

We are now ready to prove the main Theorem of this section:
\begin{theorem}
The maximum temporal flow in $\big( G(L)=(V,E,L),\source, \sink, c \big)$ is equal to the minimum capacity (minimal) temporal cut.
\end{theorem}
\begin{proof}
By Lemma~\ref{lem:steg_flow}, the maximum temporal flow in $G(L)$ is equal to the maximum flow rate from $\source$ to $\sink$ in $\mathrm{TEG}(L)$ and in $\mathrm{STEG}(L)$. But in $\mathrm{STEG}(L)$, the maximum $\source \mhyphen \sink$ flow rate is equal to the minimum $\source \mhyphen \sink$ cut~\cite{ford-fulkerson}. Now, by Lemma~\ref{lem:min_cut}, this cut is also equal in capacity to the minimum capacity $\source \mhyphen \sink$ cut in $\mathrm{TEG}(L)$. But any minimum capacity cut in $\mathrm{TEG}(L)$ is only using crossing edges and thus corresponds to a temporal cut in $G(L)$, of the same capacity (since the removal of the respective time-edges leaves no $\source \to \sink$ journey in $G(L)$).
\end{proof}

It is also easy to see that:
\begin{lemma}\label{lem:path_num}
Any static flow rate algorithm A that computes the maximum flow in a static, directed, $\source \mhyphen \sink$ network $G$ of $n$ vertices and $m$ edges in time $T(n,m)$, also computes the maximum temporal flow in a $\big( G(L)=(V,E,L),\source, \sink, c \big)$ temporal flow network in time $T(n',m')$, where $n' \leq n + 2|E_L|$ and $m' \leq n+ 3|E_L|$.
\end{lemma}
\begin{proof}
We run A on the static network $\mathrm{STEG}(L)$ of $n'$ vertices and $m'$ edges. Note that $\mathrm{STEG}(L)$ is, by construction, acyclic.
\end{proof}

\begin{note*}
In contrast to all the dynamic flows literature, our simplified time-extended network has size (number of nodes and edges) \emph{linear} on the input size of $G(L)$, and \emph{not exponential}.
\end{note*}

The following is a direct corollary of the construction of the Simplified Time-Extended Graph and shows that any temporal flow from $\source$ to $\sink$ (in temporal flow networks with unbounded node buffers) can be decomposed into temporal flows on some $\source \to \sink$ journeys.

\begin{corollary}[Journeys flow decomposition]\label{cor:decomp}
Let $\big( G(L)=(V,E,L),\source, \sink, c \big)$ be a temporal flow network on a directed graph $G$. Let $f$ be a temporal flow in $G(L)$ ($f$ is given by the values of $f(e,l)$ for the time-edges $(e,l) \in E_L$). Then, there is a collection of $\source \to \sink$ journeys $j_1,j_2,\ldots,j_k$ such that:
\begin{enumerate}
\item $k \leq |E_L|$
\item $v(f)=v(f_1)+ \ldots v(f_k)$
\item $f_i$ sends positive flow only on the time-edges of $j_i$
\end{enumerate}
\end{corollary}

\section{Mixed Temporal Networks and their hardness}

Mixed temporal networks of the form $G(E_1,E_2,\alpha)$ (see Definition~\ref{def:mixed}) can model practical cases, where some edge availabilities are exactly specified, while some other edge availabilities are randomly chosen (due to security reasons, faults, etc.); for example, in a water network, one may have planned disruptions for maintenance in some water pipes, but unplanned (random) disruptions in some others. With some edges being available at random times, the value of the maximum temporal flow (until time $\alpha$) now becomes a random variable.

In this section, we focus our attention to temporal networks that either have all their labels chosen uniformly at random, or are (fully) mixed.

\subsection{Temporal Networks with random availabilities that are flow cutters}\label{sec:random_journeys}

We study here a special case of the mixed temporal networks $G(E_1,E_2,\alpha)$, where $E_1= \emptyset$, i.e., \emph{all} the edges in the network become available at random time instances. We partially characterise such networks that eliminate the flow that arrives at $\sink$.

Let $G=(V,E)$ be a directed graph of $n$ vertices with a distinguished source, $\source$, and a distinguished sink, $\sink$. Suppose that each edge $e \in E$ is available only at a \emph{unique} moment in time (i.e., day) \emph{selected uniformly at random} from the set $\{1,2, \ldots, \alpha\}$, for some even\footnote{We choose an even integer to simplify the calculations. However, with careful adjustments to the calculations, the results would still hold for an arbitrary integer.} integer $\alpha\geq 1$; suppose also that the selections of the edges' labels are independent. Let us call such a network a Temporal Network with unique random availabilities of edges, and denote it by $\mathrm{URTN}(\alpha)$.

\begin{lemma}\label{lem:factorial_jour}
Let $P_k$ be a directed $\source \to \sink$ path of length $k$ in $G$. Then, $P_k$ becomes a journey in $\mathrm{URTN}(\alpha)$ with probability at most $\frac{1}{k!}$.
\end{lemma}
\begin{proof}

For a particular $\source \to \sink$ path $P_k$ of length $k$, let $\mathcal{E}$ be the event that ``$P_k$ is a journey'', $\mathcal{D}$ be the event that ``all $k$ labels on $P_k$ are different'' and $\mathcal{S}$ be the event that ``at least $2$ out of the $k$ labels on $P_k$ are equal''. Then, we have:
\begin{eqnarray*}
Pr[\mathcal{E}] & =    & Pr[\mathcal{E} | \mathcal{D}] \cdot Pr[\mathcal{D}]  +  Pr[\mathcal{E} | \mathcal{S}] \cdot Pr[\mathcal{S}] \\
														 & =    & Pr[\mathcal{E} | \mathcal{D}] \cdot Pr[\mathcal{D}]\\
														 & \leq & Pr[\mathcal{E} | \mathcal{D}]
\end{eqnarray*}
Now, each particular \emph{set} of $k$ different labels in the edges of $P_k$ is equiprobable. But for each such set, all permutations of the $k$ labels are equiprobable and only one is a journey, i.e., has increasing order of labels. Therefore:
\[  Pr[P_k \text{ is a journey}] \leq \frac{1}{k!} \text{.} \]
\end{proof}

Now, consider directed graphs as described above, in which the distance from $\source$ to $\sink$ is at least $c\log{n}$, for a constant integer $c>2$; so any directed $\source \to \sink $ path has at least $c\log{n}$ edges. Let us call such graphs `` $c$-long $\source \to \sink$ graphs'' or simply $c$-long. A $c$-long $\source \to \sink$ graph is called \emph{thin} if the number of simple directed $\source \to \sink$ paths is at most $n^\beta$, for some constant $\beta$.

\begin{lemma}
Consider a $URTN(\alpha)$ with an underlying graph $G$ being any particular $c$-long and thin digraph. Then, the probability that the amount of flow from $\source$ arriving at $\sink$ is positive tends to zero as $n$ tends to $+\infty$.
\end{lemma}
\begin{proof}
The event $E_1 = $``at least one $\source \to \sink$ path is a journey in $URTN(\alpha)$'' is a prerequisite for a positive flow from $\source$ arriving at $\sink$. So, 
\begin{eqnarray}
\notag Pr[\text{flow arriving at } \sink > 0] &=& Pr[\exists \source \to \sink \text{ simple path in } G \text{ which is a journey}]\\
\notag                                   & \leq & n^\beta Pr[\text{any specific simple path in } G \text{ is a journey}]  \\
	                                 & \leq & n^\beta \frac{1}{(c\log{n})!} \text{,}
\end{eqnarray}\label{eqn:my_eq}
by Lemma~\ref{lem:factorial_jour} and since every $\source \to \sink$ path in $G$ has length at least $c\log{n}$. It holds that $c! \geq \left(\frac{c}{2} \right)^\frac{c}{2}$ and that $(\log{n})! \geq \left( \frac{\log{n}}{2}\right)^\frac{\log{n}}{2}$. Therefore, relation~\ref{eqn:my_eq} becomes:

\[Pr[\text{flow arriving at } \sink > 0] \leq \frac{1}{\left(\frac{c}{2} \right)^\frac{c}{2}} \cdot \frac{n^\beta \sqrt{n}}{(\log{n})^\frac{\log{n}}{2}}\]
But, $n^\beta \sqrt{n} = o(\log{n})^\frac{\log{n}}{2}$ for $n$ large enough, so the Lemma holds.
\end{proof}

Randomly labelled $c$-long and thin graphs is not the only case of temporal networks that disallows flow to arrive to $\sink$ asymptotically almost surely.
\begin{definition}
A cut $C$ in a (traditional) flow network $G$ is a set of edges, the removal of which from the network leaves no directed $\source \to \sink$ paths in $G$.
\end{definition}
\begin{definition}
A cut $C_1$ \emph{precedes} a cut $C_2$ in a flow network $G$ (denoted by $C_1 \to C_2$) if any directed $\source \to \sink$ path that goes through an edge in $C_1$ must also later go through an edge in $C_2$.
\end{definition}
\begin{definition}[Multiblock graphs]
A flow network is called a $(c,d)$-multiblock graph if it has at least $c\log{n}$ disjoint cuts $C_1, \ldots, C_{c\log{n}}$ such that $C_i \to C_{i+1}, ~i=1, \ldots, c\log{n}-1$, and for all $i=1, \ldots, c\log{n}$, $|C_i| \leq d$, for some constants $c,d>2$.
\end{definition}

Note that $(c,d)$-multiblocks and ($c$-long,thin)-graphs are two different graph classes. Figure~\ref{fig:c2multiblock} shows a $(c,2)$-multiblock of $n = c\sqrt{k}+2, ~k \in\mathbb{N}$, vertices which is not thin.
\begin{figure}[htb]
\centering\includegraphics[scale=0.6]{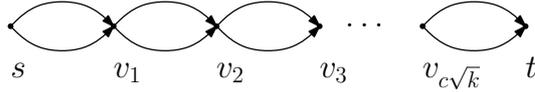}
\caption{A $(c,2)$-multiblock which is not thin.}
\label{fig:c2multiblock}
\end{figure}

\begin{lemma}
Consider a $URTN(\alpha)$ with an underlying graph $G$ being any particular $(c,d)$-multiblock. Then, the probability that the amount of flow from $\source$ arriving at $\sink$ is positive tends to zero as $n$ tends to $+\infty$.
\end{lemma}
\begin{proof}
For positive flow to arrive to $\sink$ starting from $\source$, it must be that if $C_i \to C_{i+1}$ then at least one edge availability in $C_{i+1}$ is larger than the smallest edge availability in $C_i$. Note that for every $C_i$, the probability that all labels in $C_i$ are at least $\frac{\alpha}{2}$ is $\left(\frac{1}{2}+\frac{1}{\alpha}\right)^{|C_i|}$, i.e., a constant. Also, for every $C_i$, the probability that all labels in $C_i$ are at most $\frac{\alpha}{2}$ is $\left(\frac{1}{2}\right)^{|C_i|}$, i.e., a constant.

Now, given a consecutive pair of cuts $C_i \to C_{i+1}$, let $E_{i,\geq}$ be the event that all labels in $C_i$ are at least $\frac{\alpha}{2}$ and $E_{i+1,\leq}$ be the event that all labels in $C_{i+1}$ are at most $\frac{\alpha}{2}$. Let $A_i$ be the conjunction of $E_{i,\geq}$ and $E_{i+1,\leq}$. It holds that:
\begin{eqnarray*}
\notag Pr[A_i] = Pr[E_{i,\geq} \wedge E_{i+1,\leq}] &=& Pr[E_{i,\geq}]\cdot Pr[E_{i+1,\leq}] \\
														 &\geq& \left( \frac{1}{2} + \frac{1}{\alpha} \right)^{|C_i|} \cdot \left( \frac{1}{2} \right)^{|C_{i+1}|} \\
																		&\geq& \left( \frac{1}{2}  \right)^{2d}.
\end{eqnarray*}

But, the conjunction of $E_{i,\geq}$ and $E_{i+1,\leq}$ implies that no flow arrives at $\sink$ starting from $\source$. Now, consider the events: $S=\{ A_1, A_3, A_5, \ldots, A_{r}\}$, where $r$ is the largest odd number that is smaller than $c \log{n}$; note that $r=\Theta(\log{n})$. Those events are independent since there is no edge overlap in any of them; therefore, the random label choices in any one consecutive pair of cuts does not affect the choices in the next pair. We have:
\begin{eqnarray*}
Pr[\text{flow arriving at } \sink > 0] &\leq& Pr[\text{all events in }S\text{ fail}] \\
																			 & =  & \prod_{A_j \in S} Pr[A_j \text{ fails}]\\
															&  \leq& \left( 1- \left( \frac{1}{2} \right)^{2d} \right)^{\Theta(\log{n})} \xrightarrow{n\rightarrow+\infty} 0.
\end{eqnarray*}
This completes the proof of the Lemma.
\end{proof}

\subsection{The complexity of computing the expected maximum temporal flow}\label{sec:mixed_net}

We consider here the following problem:
\begin{prob*}[Expected Maximum Temporal Flow]
What is the time complexity of computing the \emph{expected value} of the maximum temporal flow, $v$, in $G(E_1,E_2,\alpha)$?
\end{prob*}

Let us recall the definition of the class of functions \pmb{\#P}:
\begin{definition}
\cite[p.441]{papadimitriou} Let $Q$ be a polynomially balanced, polynomial-time decidable binary relation. The \emph{counting problem} associated with $Q$ is: Given $x$, how many $y$ are there such that $(x,y) \in Q$? \pmb{\#P} is the class of all counting problems associated with polynomially balanced polynomial-time decidable functions.
\end{definition}
Loosely speaking, a problem is said to be \pmb{\#P}-hard if a polynomial-time algorithm for it implies that  \pmb{\#P} $=$ \pmb{FP}, where \pmb{FP} is the set of functions from $\{0,1\}^*$ to $\{0,1\}^*$ computable by a deterministic polynomial-time Turing machine\footnote{$\{0,1\}^* = \cup_{n \geq 0} \{0,1\}^n$, where $\{0,1\}^n$ is the set of all strings (of bits $0,1$) of length $n$}. For a more formal definition, see~\cite{papadimitriou}.

We now show the following:
\begin{lemma}\label{lem:sharp_complete}
Given an integer $C>0$, it is \pmb{\#P}-hard to compute the probability that the maximum flow value $v$ in $G(E_1,E_2,\alpha)$ is at most $C$, $Pr[v \leq C]$.
\end{lemma}
\begin{proof}
Recall that if $J=\{w_1,\ldots, w_n\}$ is a set of $n$ positive integer weights and we are given an integer $C \geq \sum_{i=1}^n \frac{w_i}{2}$, then the problem of computing the number, $T$, of subsets of $J$ with total weight at most $C$ is \pmb{\#P}-hard, because it is equivalent to counting the number of feasible solutions of the corresponding KNAPSACK instance~\cite{papadimitriou}.

Consider now the temporal flow network of Figure~\ref{fig:mixed_net'} where there are $n$ directed disjoint two-edge paths from $\source$ to $\sink$. For the path with edges $e_i,e_i'$, via vertex $v_i$, the capacity of $e_i$ is $w_i$ and the capacity of $e_i'$ is $w_i' \geq w_i$. In this network, $E_1= \emptyset$ and $E_2=E$, i.e., the availabilities of every edge are chosen independently and uniformly at random from $\{1,\ldots, \alpha\}$. Also, assume that each edge selects a \emph{single} random label.

\begin{figure}[tbh]
\centering
\includegraphics[scale=0.45]{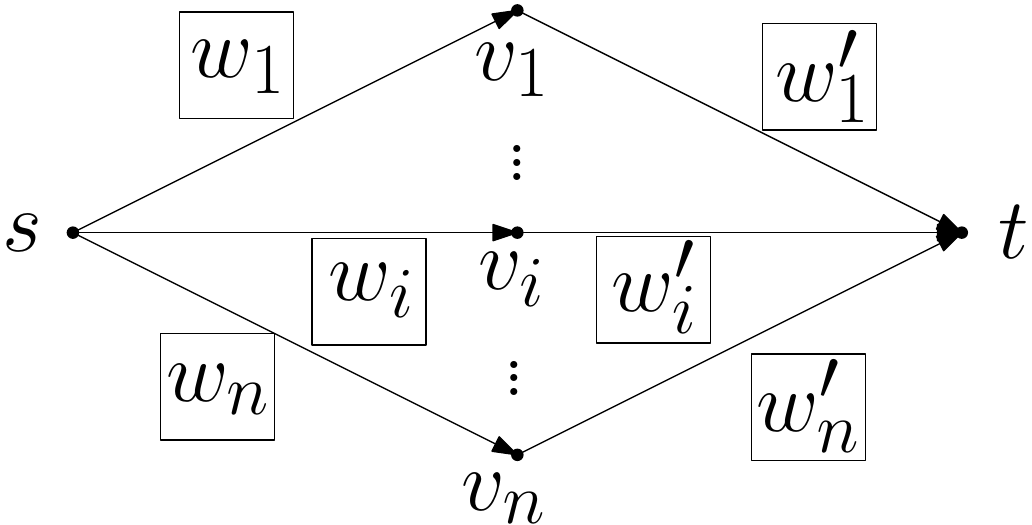}
\caption{The network structure we consider}
\label{fig:mixed_net'}
\end{figure}

Clearly, the value of the maximum temporal flow from $\source$ to $\sink$ until time $\alpha+tt$ is the sum of $n$ random variables $Y_i,~i=1,\ldots, n$, where $Y_i$ is the value of the flow through the $i^{th}$ path. $Y_i$ is, then, $w_i$ with probability $p_i= \frac{1}{2}- \frac{1}{2\alpha}$, which is equal to the probability that the label $l_{e_i}$ is smaller than the label $l_{e_i'}$, so that the path $(e_i,e_i')$ is a journey, and is zero otherwise. Then, $v=Y_1+ \ldots + Y_n$ and it holds that $Pr[v\leq C] = Pr[\sum_{i=1}^n Y_i \leq C]$.

Now, let $J_k$ be the set of all vectors, $(\rho_1,\ldots, \rho_n)$, of $n$ entries/weights in total, such that each $\rho_i$ is either $0$ or the corresponding $w_i$, and there are exactly $k$ positive entries in the vector.
Let $\vec{g} = (g_1, \ldots, g_n)$ be a specific assignment of weights to $Y_1,\ldots, Y_n$, respectively, i.e., $g_i=w_i$ with probability $\frac{1}{2}-\frac{1}{2\alpha}$ and, otherwise, $g_i=0$; notice that $\vec{g} \in J_k$, for some $k\in \{0,\ldots,n\}$. Then,
\begin{eqnarray}\label{eq:mixed}
\notag  Pr[v \leq C] & =& Pr[\sum_{i=1}^n Y_i \leq C]\\
				& =& \sum_{\vec{g}} Pr[Y_i = g_i, ~\forall i=1,\ldots,n] \cdot x(\vec{g}) \text{,}
\end{eqnarray}
where:
\[x(\vec{g}) = 
\begin{cases} 1 & \mbox{, if } \sum_{i=1}^n g_i \leq C \\     0 & \mbox{, otherwise.} \end{cases}
\]

For each particular $\vec{g}$ with exactly $k$ positive weights, the probability that it occurs is $\big( \frac{1}{2} - \frac{1}{2\alpha} \big)^k \big(\frac{1}{2} + \frac{1}{2\alpha} \big)^{n-k}$. So, from Equation~\ref{eq:mixed} we get:

\begin{eqnarray}\label{eq:mixed2}
\notag Pr[v\leq C] &=& \sum_{k=0}^n \sum_{\vec{g} \in J_k} x(\vec{g}) \big( \frac{1}{2} - \frac{1}{2\alpha} \big)^k \big(\frac{1}{2} + \frac{1}{2\alpha} \big)^{n-k} \\
&=& \big(\frac{1}{2} + \frac{1}{2\alpha} \big)^n \sum_{k=0}^n \sum_{\vec{g} \in J_k} x(\vec{g}) \left( \frac{\frac{1}{2} - \frac{1}{2\alpha}}{\frac{1}{2} + \frac{1}{2\alpha}} \right)^k
\end{eqnarray}

The following holds (using Bernoulli's inequality):
\begin{equation}\label{eq:mixed3}
1 \geq  \left( \frac{\frac{1}{2} - \frac{1}{2\alpha}}{\frac{1}{2} + \frac{1}{2\alpha}} \right)^k  \geq \left( \frac{\frac{1}{2} - \frac{1}{2\alpha}}{\frac{1}{2} + \frac{1}{2\alpha}} \right)^n 
= \left( \frac{\alpha-1}{\alpha+1} \right)^n = \big( 1-\frac{2}{\alpha+1} \big)^n \geq  1-\frac{2n}{\alpha+1}
\end{equation}

Let $T= \sum_{k=0}^n \sum_{\vec{g} \in J_k} x(\vec{g})$ and note that $T$ is exactly the number of subsets of $J=\{w_1, \ldots, w_n\}$ with total weight at most $C$. Then, we get from Equation~\ref{eq:mixed2} and Relation~\ref{eq:mixed3}:
\begin{equation*}
\begin{array}{rcccll}
\big(\frac{1}{2} + \frac{1}{2\alpha} \big)^n \big( 1-\frac{2n}{\alpha+1} \big) T  &\leq& Pr[v\leq C] &\leq& \big(\frac{1}{2} + \frac{1}{2\alpha} \big)^n T & \Leftrightarrow\\[0.5em]
\big( 1-\frac{2n}{\alpha+1} \big) T  &\leq&     Pr[v\leq C] \frac{1}{\big(\frac{1}{2} + \frac{1}{2\alpha} \big)^n} & \leq & T  &\Leftrightarrow\\[0.5em]
T-\frac{2nT}{\alpha+1}               &\leq&     \frac{Pr[v\leq C] }{\big(\frac{1}{2} + \frac{1}{2\alpha} \big)^n}      & \leq& T &
\end{array}
\end{equation*}

Now, assume that $\alpha+1>2nT$; we can guarantee that by selecting $\alpha$ to be, for example, $2^n$, or larger. Then, $0< \frac{2nT}{\alpha+1}<1 $. Let $\varepsilon=\frac{2nT}{\alpha+1}$. Then, we get:
\[  T-\varepsilon \leq \frac{Pr[v\leq C] }{\big(\frac{1}{2} + \frac{1}{2\alpha} \big)^n} \leq T  \]

Note that $\big(\frac{1}{2} + \frac{1}{2\alpha} \big)^n$ can be represented by a polynomial in $n$ number of bits and can be computed in polynomial time.

If we had a polynomial-time algorithm, $A$, to exactly compute $Pr[v\leq C]$ for any $C$ and $\alpha$, then we could exactly compute (also in polynomial time) a number between $T-\varepsilon$ and $T$, for $0<\varepsilon<1$. But, this determines $T$ exactly. So, such an algorithm $A$ would solve a \pmb{\#P}-hard problem in polynomial time.
\end{proof}

\begin{remark*}
If each of the random variables $Y_i$ was of the form $Y_i=w_i$ with probability $p_i=\frac{1}{2}$, and zero otherwise, then the reduction to the KNAPSACK problem would be immediate~\cite{fotakis, kleinberg}. However, the possibility of ties in the various $l_{e_i}$ and $l_{e_i'}$s excludes the respective journeys and the reduction does not carry out immediately. 
\end{remark*}

Now, given a mixed temporal network $G(E_1,E_2, \alpha)$, let $v$ be the random variable representing the maximum temporal flow in $G$.

\begin{definition}
The truncated by $B$ expected maximum temporal flow of $G(E_1,E_2, \alpha)$, denoted by $E[v,B]$, is defined as:
\[
E[v,B] = \sum_{i=1}^B i Pr[v=i]
\]
Clearly, it is $E[v]=E[v, +\infty]$.
\end{definition}

We are now ready to prove the main theorem of this section:
\begin{theorem}
It is \pmb{\#P}-hard to compute the expected maximum truncated Temporal Flow in a Mixed Temporal Network $G(E_1,E_2,\alpha)$.
\end{theorem}
\begin{proof}

Consider the single-labelled mixed temporal network $G(E_1,E_2,\alpha)$ of Figure~\ref{fig:mixed_truncated}, in which $\source$ has $n$ outgoing disjoint directed paths of two edges $e_i,e_i'$ to a node $t_1$, and then there is an edge from $t_1$ to $\sink$. The capacity of each edge $(\source,v_i)~,i=1,\ldots,n$, is $w_i$, the capacity of each edge $(v_i, t_1)~,i=1,\ldots,n$, is $w_i' \geq w_i$, and the capacity of the edge $(t_1, \sink)$ is an integer $B$ such that  $\frac{1}{2} \sum_{i=1}^n w_i <B< \sum_{i=1}^n w_i$. The \emph{unique label} of edge $(t_1,\sink)$ is some $b\in \mathbb{N},~b>\alpha$, where $\alpha$ is the maximum possible label that the other edges may select; in particular, each of the edges $(\source,v_i), (v_i, t_1) ~,i=1,\ldots,n$ receives a unique random label drawn uniformly and independently from $\{1, \ldots, \alpha\}$.

\begin{figure}[tbh]
\centering
\includegraphics[scale=0.45]{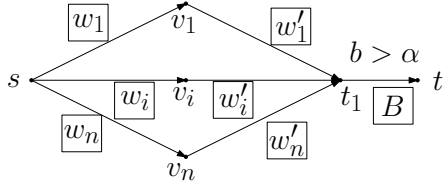}
\caption{A $G(E_1,E_2,\alpha)$ where $E_1=\{(t_1,\sink)\}$ with $l_{(t_1,\sink)}=b>\alpha$.}
\label{fig:mixed_truncated}
\end{figure}

Clearly, the maximum temporal flow from $\source$ to $\sink$ until time $b$ is is $v'=B$, if $v=\sum_{i=1}^n Y_i>B$, and is $v' = v =\sum_{i=1}^n Y_i$, otherwise; here $Y_i,~i=1, \ldots,n$, is the random variable representing the flow passing from $\sink$ to $t_1$ via $v_i$ in the time until $\alpha$.

So, if $E[v']$ is the expected value of $v'$, we have:
\begin{eqnarray}\label{eq:final}
\notag E[v']  & = & \sum_{i=0}^B i Pr[v=i] + B \cdot Pr[v>B] \\
  	        	& = & E[v,B] + B\big( 1-Pr[v\leq B] \big)
\end{eqnarray}
So, if we had a polynomial-time algorithm that could compute truncated expected maximum temporal flow values in mixed temporal networks, then we could compute $E[v']$ and $E[v,B]$; we could then solve Equation~\ref{eq:final} for $Pr[v \leq B]$ and, thus, compute it in polynomial time. But to compute $Pr[v \leq B]$ is \pmb{\#P}-hard by Lemma~\ref{lem:sharp_complete}.
\end{proof}

\section{Conclusions}
We defined and studied here for the first time flows in temporal networks. Our intuitive characterization of temporal cuts for networks with unbounded buffers may lead to fast algorithmic techniques (perhaps by sampling) for computing a minimum cut in such a network. We also considered random availabilities in some of the edges of our networks (mixed temporal networks). An interesting open problem is the existence of a FPTAS for the expected maximum flow value in mixed temporal networks. Another type of dynamic graphs that would be interesting to investigate with respect to the complexity of the maximum flow (by some day $d \in \mathbb{N}$) problem is that of periodic temporal graphs. These are graphs each edge $e$ of which appears every $x_e$ days; $x_e$ is what we call the “edge period”. The maximum flow from $\source$ to $\sink$ would then, in general, increase when we increase the day $d$ by which we wish to compute the flow that arrives at $\sink$. It appears that the problem would require a different approach than the one presented here, that would also take into account the different edge periods.

\bibliographystyle{abbrv}
\bibliography{flows}   
\end{document}